\documentclass{article}%
\usepackage{amsfonts}
\usepackage{amsmath}
\usepackage{amssymb}
\usepackage{graphicx}%
\newtheorem{theorem}{Theorem}

\newtheorem{proposition}[theorem]{Proposition}
\newtheorem{remark}[theorem]{Remark}

\newenvironment{proof}[1][Proof]{\noindent\textbf{#1.} }{\ \rule{0.5em}{0.5em}}
\begin{document}

\title{On Bers generating functions for first order systems of mathematical physics}
\author{Vladislav V. Kravchenko$^{1}$, Marco P. \ Ramirez T.$^{2}$\\$^{1}${\small Departamento de Matem\'{a}ticas, CINVESTAV del IPN, Unidad
Quer\'{e}taro,}\\{\small Libramiento Norponiente No. 2000, Fracc. Real de Juriquilla, }\\{\small Quer\'{e}taro, Qro. C.P. 76230 MEXICO e-mail:
vkravchenko@qro.cinvestav.mx}\\$^{2}${\small Escuela de Ingenier\'{\i}a de la Universidad La Salle,}\\{\small \ Benjam\'{\i}n Franklin No. 47, Col. Condesa, C.P. 06140, M\'{e}xico
D.F.}\\{\small \ e-mail: mramirez@lci.ulsa.mx}}
\maketitle

\begin{abstract}
Considering one of the fundamental notions of Bers' theory of pseudoanalytic
functions the generating pair via an intertwining relation we introduce its
generalization for biquaternionic equations corresponding to different
first-order systems of mathematical physics with variable coefficients. We
show that the knowledge of a generating set of solutions of a system allows
one to obtain its different form analogous to the complex equation describing
pseudoanalytic functions of the second kind and opens the way for new results
and applications of pseudoanalytic function theory. As one of the examples the
Maxwell system for an inhomogeneous medium is considered, and as one of the
consequences of the introduced approach we find a relation between the
time-dependent one-dimensional Maxwell system and hyperbolic pseudoanalytic
functions and obtain an infinite system of solutions of the Maxwell system.
Other considered examples are the system describing force-free magnetic fields
and the Dirac system from relativistic quantum mechanics.

\end{abstract}

\section{Introduction}

Bers' theory of pseudoanalytic functions mainly created in fifties of the last
century \cite{bers} offers interesting and still not fully explored tools for
studying and solving linear elliptic equations in the plane. Recent advances
in the theory and its applications \cite{APFT} show that some abstract
constructions proposed by Bers can be made completely explicit and applicable
to important equations of mathematical physics. For example, pseudoanalytic
formal powers introduced and studied by Bers and later on by\ many other
mathematicians had been obtained only in some very special situations which
represented a considerable obstacle for a further development of
pseudoanalytic function theory. This obstacle has been substantially
diminished in a recent work reported in \cite{APFT} due to the fact that there
has been found a method for constructing formal powers explicitly in a much
more general situation. Moreover, it was shown that this construction offers a
tool for calculating explicitly complete systems of solutions of linear
elliptic second-order equations in the plane. This progress together with some
other recent developments posed more open questions concerning pseudoanalytic
function theory, its generalizations and applications. For example, a
possibility to develop a hyperbolic pseudoanalytic function theory with
applications to hyperbolic equations of mathematical physics was explored in
\cite{KRT} (see also \cite{APFT} and \cite{KKT}). In the present work we use
some of the results of \cite{KRT} for obtaining an infinite system of
solutions of the one-dimensional time-dependent Maxwell system. Another
related important open question is the development of pseudoanalytic function
theory in higher dimensions. This was analized in a number of papers (see
\cite{malonek}, \cite{KrJPhys06}, \cite{berglez} and \cite{APFT}), and it is
clear that in this direction the development is barely starting.

The aim of the present paper is to show that some fundamental ideas of
pseudoanalytic function theory are valid in a general situation and applicable
to linear systems of mathematical physics both elliptic and hyperbolic. The
starting point of Bers' theory is the concept of a generating pair which in a
sense means a substitution of a pair of \textquotedblleft rectilinear
elements\textquotedblright\ of the plane $1$ and $i$ by a pair of quite
arbitrary \textquotedblleft curvilinear elements\textquotedblright\ -- a pair
of complex functions $F$ and $G$ which only should enjoy the property of
independence in the sense that any complex function $w$ can be represented in
the form $w=\varphi F+\psi G$ where $\varphi$ and $\psi$ are real-valued
functions. Beginning with generalizations of the first definitions from
analytic function theory like the derivative, Bers shows that behind a
generating pair there is always a corresponding generalized Cauchy-Riemann
system which is usually called the Vekua or Carleman-Vekua equation. The
knowledge of a generating pair for a Vekua equation allows one to represent it
in another form which is the equation for pseudoanalytic functions of the
second kind. This form is very convenient for introducing a simple formula for
calculation of the $(F,G)$-derivative in the sense of Bers and of the
corresponding antiderivative and in fact represents a cornerstone of all
further constructions of pseudoanalytic function theory including formal powers.

In this work we show that the concept of a generating set of functions in the
sense of Bers is much more universal and can be introduced in relation with
first-order systems of mathematical physics with the aid of hypercomplex
algebraic tools. This allows one to obtain another form of a corresponding
system analogous to that describing pseudoanalytic functions of the second
kind. In order to generalize the concept of a generating pair it resulted to
be fruitful to develop a slightly different approach to this concept via a
certain intertwining relation. We consider it in the next section. Another
tool implemented in this paper is the algebra of biquaternions. The related
notations are introduced also in section \ref{SectPrelim}.

We consider only few examples in this paper chosen in such a way that from one
side it becomes clear that our approach is general and is applicable to a wide
variety of systems and from the other it can be seen that for each particular
physically meaningful system some special interesting phenomena may occur.
Thus, we consider the time-dependent Maxwell system for inhomogeneous media
(section 3), the system describing force-free magnetic fields (section 4) and
the Dirac equation (section 5). In the case of the Maxwell system we introduce
a generating set of solutions, with its aid we obtain the Maxwell system in
the form of an equation for pseudoanalytic functions of the second kind. This
allows us to find a relation of the Maxwell system in the one-dimensional case
with the hyperbolic pseudoanalytic function theory developed in \cite{KRT}. As
a direct consequence of this relation we obtain an infinite system of
solutions of the Maxwell equations.

In the case of force-free magnetic fields we show that in fact one exact
solution is sufficient to obtain a corresponding generating quartet and to
write down the second-kind equation. This interesting phenomenon leads to an
observation regarding the quotients of solutions. Namely, we obtain a
differential equation satisfied by the quotients of solutions. This result
seems to be new even in the case of monogenic functions (a special case when
the proportionality factor in the considered system vanishes identically).
Finally, in section 5 we show that the concept of a generating set is
applicable to the Dirac system with electromagnetic and scalar potentials, and
in this case as well the system can be written in another form corresponding
to pseudoanalytic functions of the second kind.

\section{Preliminaries\label{SectPrelim}}

\subsection{Biquaternions}

We will denote the algebra of biquaternions or complex quaternions by
$\mathbb{H(C)}$ with the standard basic quaternionic units denoted by
$e_{0}=1$, $e_{1},e_{2}$ and $e_{3}$. The complex imaginary unit is denoted by
$i$ as usual. The set of purely vectorial quaternions $q=\mathbf{q}$ is
identified with the set of three-dimensional vectors.

The quaternionic conjugation of a biquaternion $q=q_{0}+\mathbf{q}$ will be
denoted as%
\[
\overline{q}=q_{0}-\mathbf{q},
\]
and by $q^{\ast}$ we denote the complex conjugation of $q$,%
\[
q^{\ast}=\operatorname{Re}q-i\operatorname{Im}q.
\]

Sometimes the following notation for the operator of multiplication from the
right-hand side will be used
\[
M^{p}q=q\cdot p.
\]

The main quaternionic differential operator introduced by Hamilton himself and
sometimes called the Moisil-Theodoresco operator is defined on continuously
differentiable biquaternion-valued functions of the real variables $x_{1}$,
$x_{2}$ and $x_{3}$ according to the rule%
\[
Dq=\sum_{k=1}^{3}e_{k}\partial_{k}q,
\]
where $\partial_{k}=\frac{\partial}{\partial x_{k}}$.\ 

\subsection{Pseudoanalytic functions}

In this subsection we introduce some basic concepts from Bers' theory of
pseudoanalytic functions \cite{bers} and give a slightly varied interpretation
of the notion of a generating pair. Precisely this different interpretation
allows us to introduce the generating sets for the considered in the
subsequent sections systems of mathematical physics.

According to \cite{bers} a pair of arbitrary continuously differentiable with
respect to the real variables $x$ and $y$ complex-valued functions $F$ and $G$
satisfying the inequality%
\begin{equation}
\operatorname{Im}(\overline{F}G)>0 \label{ps01}%
\end{equation}
in a domain $\Omega\subset\mathbb{C}$ is called a generating pair. This
inequality means that $F$ and $G$ are independent in the sense that any
complex function $W$ defined in $\Omega$ can be expressed in the form%
\[
W=\varphi F+\psi G
\]
where $\varphi$ and $\psi$ are real-valued functions.

For a fixed generating pair and in the case when $\varphi$ and $\psi$ are
continuously differentiable one can define the $(F,G)$-derivative of the
function $W$ in the following way%
\[
\mathring{W}=F~\partial_{z}\varphi+G~\partial_{z}\psi,
\]
where $\partial_{z}=\frac{1}{2}(\partial_{x}-i\partial_{y})$. The derivative
$\mathring{W}$ exists iff the equality%
\begin{equation}
F~\partial_{\overline{z}}\varphi+G~\partial_{\overline{z}}\psi=0\label{ps02}%
\end{equation}
holds. Here $\partial_{\overline{z}}=\frac{1}{2}(\partial_{x}+i\partial_{y})$.

Introducing the notation%
\[
a=-\frac{\overline{F}\partial_{\overline{z}}G-\overline{G}\partial
_{\overline{z}}F}{F\overline{G}-\overline{F}G},\text{ \ \ }b=\frac
{F\partial_{\overline{z}}G-G\partial_{\overline{z}}F}{F\overline{G}%
-\overline{F}G},
\]
equation (\ref{ps02}) can be written in the form of a Vekua equation%
\begin{equation}
\partial_{\overline{z}}W-aW-b\overline{W}=0. \label{ps03}%
\end{equation}
Functions $a$ and $b$ are known as characteristic coefficients of the
generating pair\emph{ }$\left(  F,G\right)  $ and solutions of (\ref{ps03})
are known as pseudoanalytic functions, or more exactly $\left(  F,G\right)
$-pseudoanalytic functions of the first kind. Solutions of the corresponding
equation (\ref{ps02}) regarded as complex-valued functions $w=\varphi+i\psi$
are called $\left(  F,G\right)  $-pseudoanalytic functions of the second kind.
Note that by construction both generating functions $F$ and $G$ are $\left(
F,G\right)  $-pseudoanalytic of the first kind.

Now let us consider the operator from Vekua equation (\ref{ps03}),
$\partial_{\overline{z}}-a-bC$ where by $C$ we denote the operator of complex
conjugation. Take an arbitrary real-valued function $\varphi$ and consider the
equality
\begin{equation}
\left(  \partial_{\overline{z}}-a-bC\right)  \left(  \varphi f\right)
=f\partial_{\overline{z}}\varphi\label{ps06}%
\end{equation}
where $f$ is some complex function. It is easy to see that this equality holds
for any real-valued $\varphi$ iff $f$ is a particular solution of
(\ref{ps03}). In order to be able to consider a general solution of
(\ref{ps03}) another particular solution, say, $g$ is needed. In this case we
can look for solutions of (\ref{ps03}) in the form $W=\varphi f+\psi g$ where
$\varphi$ and $\psi$ are real-valued functions and $f$ and $g$ are particular
solutions of (\ref{ps03}) if only $f$ and $g$ are independent in the sense
explained above. Thus we arrive at the concept of a generating pair for a
Vekua equation via the intertwining relation (\ref{ps06}). We have then
\[
\left(  \partial_{\overline{z}}-a-bC\right)  \left(  \varphi f+\psi g\right)
=f\partial_{\overline{z}}\varphi+g\partial_{\overline{z}}\psi
\]
which in particular gives us a relation between equations (\ref{ps03}) and
(\ref{ps02}).

\section{Generating sets of solutions for the Maxwell system in inhomogeneous
media}

Let us consider the Maxwell equations for inhomogeneous media%
\begin{equation}
\operatorname{rot}\mathbf{H}=\varepsilon\partial_{t}\mathbf{E}+\mathbf{j,}%
\label{Min1}%
\end{equation}%
\begin{equation}
\operatorname{rot}\mathbf{E}=-\mu\partial_{t}\mathbf{H},\label{Min2}%
\end{equation}%
\begin{equation}
\operatorname{div}(\varepsilon\mathbf{E})=\mathbf{\rho},\label{Min3}%
\end{equation}%
\begin{equation}
\operatorname{div}\mathbf{(}\mu\mathbf{H})=0.\label{Min4}%
\end{equation}
\ Here $\varepsilon$ and $\mu$ are real-valued functions of coordinates,
$\mathbf{E}$ and $\mathbf{H}$ are real-valued vector fields depending on $t$
and spatial variables, the real-valued scalar function $\mathbf{\rho}$ and the
real vector function $\mathbf{j}$ characterize the distribution of sources of
the electromagnetic field.

The wave propagation velocity will be denoted by $c=\frac{1}{\sqrt
{\varepsilon\mu}}$, the refraction index by $n=\sqrt{\varepsilon\mu}$ and the
intrinsic impedance of the medium by $Z=\sqrt{\frac{\mu}{\varepsilon}}$. As
was shown in \cite{k-inhom}, \cite{kbook}, introducing the notations%
\[
\mathbf{c}=\frac{\operatorname*{grad}\sqrt{c}}{\sqrt{c}},\text{ \ }%
\mathbf{Z}=\frac{\operatorname*{grad}\sqrt{Z}}{\sqrt{Z}}\quad\text{and
}\mathbf{V}=\sqrt{\varepsilon}\mathbf{E}+i\sqrt{\mu}\mathbf{H}%
\]
one can rewrite system (\ref{Min1})-(\ref{Min4}) in the form of a single
biquaternionic equation%
\[
(\frac{1}{c}\partial_{t}+iD)\mathbf{V}-M^{i\mathbf{c}}\mathbf{V}%
-M^{i\mathbf{Z}}\mathbf{V}^{\ast}=-(\sqrt{\mu}\mathbf{j}+\frac{i\mathbf{\rho}%
}{\sqrt{\varepsilon}})
\]
which in a sourceless situation becomes%
\begin{equation}
(\frac{1}{c}\partial_{t}+iD)\mathbf{V}-M^{i\mathbf{c}}\mathbf{V}%
-M^{i\mathbf{Z}}\mathbf{V}^{\ast}=0.\label{MaxMain}%
\end{equation}
Let $\varphi$ be a real-valued function. Then it is easy to see that the
equality
\begin{equation}
(\frac{1}{c}\partial_{t}+iD-M^{i\mathbf{c}}-M^{i\mathbf{Z}}C)[\varphi
\mathbf{V}]=(\frac{1}{c}\partial_{t}+iD)[\varphi]\cdot\mathbf{V}%
\label{Mintertwine}%
\end{equation}
holds iff $\mathbf{V}$ is a solution of (\ref{MaxMain}).

Assume that $\left\{  \mathbf{V}_{1},\ldots,\mathbf{V}_{6}\right\}  $ are
solutions of (\ref{MaxMain}) independent in the sense that for any complex
vector function $\mathbf{V}$ there exist real valued functions $\varphi_{k}$,
$k=1,2,\ldots,6$ such that $\mathbf{V=}\sum_{k=1}^{6}\varphi_{k}\mathbf{V}%
_{k}$. This can be easily written as a condition on a corresponding
determinant of a matrix formed by components of $\mathbf{V}_{k}$. Then due to
(\ref{Mintertwine}) we have that $\mathbf{V=}\sum_{k=1}^{6}\varphi
_{k}\mathbf{V}_{k}$ is a solution of (\ref{MaxMain}) if and only if
\begin{equation}
\sum_{k=1}^{6}(\frac{1}{c}\partial_{t}+iD)[\varphi_{k}]\cdot\mathbf{V}_{k}=0.
\label{MaxSecKind}%
\end{equation}
This equation for real-valued functions $\varphi_{k}$, $k=1,2,\ldots,6$ is a
Bers' equation for Maxwell pseudoanalytic functions of the second kind.

Notice that in a frequently encountered in practice case $\mu
=\operatorname*{Const}$ we have that $\mathbf{c}=\mathbf{Z}$ and
(\ref{MaxMain}) turns into the equation
\begin{equation}
(\frac{1}{c}\partial_{t}+iD)\mathbf{V}-i(\mathbf{V}+\mathbf{V}^{\ast
})\mathbf{c}=0, \label{MaxNonMagn}%
\end{equation}
for which it is easy to propose a triplet of independent solutions in the form%
\[
\mathbf{V}_{4}=ie_{1},\quad\mathbf{V}_{5}=ie_{2},\quad\mathbf{V}_{6}=ie_{3}%
\]
corresponding to a constant magnetic field. Thus, in order to rewrite
Maxwell's system in the form (\ref{MaxSecKind}) it is sufficient to find
another triplet of solutions. In some cases this can be done relatively easy.
Let us consider a stratified medium, that is $\varepsilon=\varepsilon(x)$ and
hence $\mathbf{c}=c_{1}(x)e_{1}=\frac{c\prime(x)}{2c(x)}e_{1}$. Then the
vectors
\[
\mathbf{V}_{1}=ce_{1},\quad\mathbf{V}_{2}=\frac{e_{2}}{c},\quad\text{and
}\mathbf{V}_{3}=\frac{e_{3}}{c}%
\]
are solutions of (\ref{MaxNonMagn}). Consequently, the Maxwell system in this
case is equivalent to equation (\ref{MaxSecKind}) (for Maxwell pseudoanalytic
functions of the second kind):%
\[
(\frac{1}{c}\partial_{t}+iD)\varphi_{1}\cdot ce_{1}+\sum_{k=2}^{3}(\frac{1}%
{c}\partial_{t}+iD)\varphi_{k}\cdot\frac{e_{k}}{c}+\sum_{k=4}^{6}(\frac{1}%
{c}\partial_{t}+iD)\varphi_{k}\cdot ie_{k-3}=0.
\]
For a better understanding of this equation as well as of equation
(\ref{MaxNonMagn}) let us consider solutions depending on $t$ and $x$ only:%
\[
(\frac{1}{c(x)}\partial_{t}+ie_{1}\partial_{x})\mathbf{V}(t,x)-i(\mathbf{V}%
(t,x)+\mathbf{V}^{\ast}(t,x))c_{1}(x)e_{1}=0.
\]
One can observe that equations for $V_{1}$ and those for $V_{2}$, $V_{3}$ are
not coupled. We have
\begin{equation}
(\frac{1}{c}\partial_{t}+ie_{1}\partial_{x})V_{1}-i(V_{1}+V_{1}^{\ast}%
)c_{1}e_{1}=0 \label{MaxOne1}%
\end{equation}
and
\begin{equation}
(\frac{1}{c}\partial_{t}+ie_{1}\partial_{x})(V_{2}e_{2}+V_{3}e_{3}%
)-i((V_{2}+V_{2}^{\ast})e_{2}+(V_{3}+V_{3}^{\ast})e_{3})c_{1}e_{1}=0.
\label{MaxOne2}%
\end{equation}
The first of these equations can be easily solved. Note that $\partial
_{t}V_{1}\equiv0$ and $\partial_{x}V_{1}-2c_{1}\operatorname*{Re}V_{1}=0$.
That is $\operatorname*{Im}V_{1}\equiv\operatorname*{Const}$ and $\partial
_{x}\operatorname*{Re}V_{1}-\frac{c\prime(x)}{c(x)}\operatorname*{Re}V_{1}=0$
which gives us a general form of the component $V_{1}$ in the case under
consideration, $V_{1}=a_{1}c(x)+ia_{2}$ where $a_{1}$ and $a_{2}$ are
arbitrary real constants.

Now let us consider equation (\ref{MaxOne2}). It can be written as the
following bicomplex equation
\[
(\frac{1}{c}\partial_{t}+ie_{1}\partial_{x})\Phi+ic_{1}e_{1}(\Phi+\Phi^{\ast
})=0
\]
for the bicomplex function $\Phi=V_{2}+V_{3}e_{1}$.

Denote by $N$ an antiderivative of the refraction index $n$ and consider the
following change of the variable $x\mapsto\xi=N(x)$. Then the function
$\Psi(t,\xi(x))=\Phi(t,x)$ as a function of the variables $t$ and $\xi$
satisfies the following equation%
\[
(\partial_{t}+ie_{1}\partial_{\xi})\Psi(t,\xi)+ie_{1}\frac{C^{\prime}(\xi
)}{2C(\xi)}(\Psi(t,\xi)+\Psi^{\ast}(t,\xi))=0
\]
where $C(\xi(x))=c(x)$. Note that introducing a new unity $j=ie_{1}$ which
obviously satisfies the equality $j^{2}=1$ and considering the bicomplex
function in the form $\Psi=\Psi_{1}+\Psi_{2}j$ where $\Psi_{1,2}%
=u_{1,2}+v_{1,2}e_{1}$ with $u_{1,2}$ and $v_{1,2}$ being real valued
functions we can rewrite the last equation as follows
\[
(\partial_{t}+j\partial_{\xi})\Psi(t,\xi)+j\frac{C^{\prime}(\xi)}{2C(\xi
)}(\Psi(t,\xi)+\Psi^{\ast}(t,\xi))=0
\]
where $\Psi^{\ast}=\Psi_{1}-\Psi_{2}j$. Finally, applying the conjugation
operator and introducing the new bicomplex function $W=\sqrt{C}\Psi^{\ast}$ we
arrive at the equation
\[
\frac{1}{2}(\partial_{\xi}-j\partial_{t})W-\frac{f^{\prime}(\xi)}{2f(\xi
)}W^{\ast}=0
\]
where $f=\sqrt{C}$, which can be written in the form of a hypebolic Vekua
equation
\begin{equation}
\partial_{\overline{z}}W-\frac{f_{\overline{z}}}{f}\overline{W}%
=0\label{BicompW}%
\end{equation}
where $\partial_{\overline{z}}=\frac{1}{2}(\partial_{\xi}-j\partial_{t})$ and
instead of an asterisk we used bar for denoting the same conjugation with
respect to $j$. This equation in the case when $W$ has values in the algebra
of hyperbolic numbers which we denote as $\mathcal{H}$, that is when
$W=W_{1}+W_{2}j$ with $W_{1}$ and $W_{2}$ being real valued, was introduced
and studied in \cite{KRT} and \cite{APFT} in relation to the Klein-Gordon
equation and in \cite{KKT} in relation to the Zakharov-Shabat system. Notice
that due to the fact that $f$ is real valued, equation (\ref{BicompW}) in fact
consists of two separate equations for $\mathcal{H}$-valued functions, that is
it reduces to a pair of equations which were considered in the previous
publications \cite{KKT}, \cite{APFT} and \cite{KRT}. In order to observe this
one needs to write the function $W$ in the form $W=w_{1}+w_{2}e_{1}$ where
$w_{1}$ and $w_{2}$ are $\mathcal{H}$-valued. Then equation (\ref{BicompW}) is
equivalent to the following pair of separate equations%
\begin{equation}
\partial_{\overline{z}}w_{1}-\frac{f_{\overline{z}}}{f}\overline{w_{1}}%
=0\quad\text{and\quad}\partial_{\overline{z}}w_{2}-\frac{f_{\overline{z}}}%
{f}\overline{w_{2}}=0.\label{w1w2}%
\end{equation}
Now gathering all the introduced transformations we have that
\[
w_{1}(t,\xi)=\sqrt{C(\xi)}\left(  \sqrt{\widetilde{\varepsilon}(\xi
)}\widetilde{E}_{2}(t,\xi)-\sqrt{\mu}\widetilde{H}_{3}(t,\xi)j\right)
\]
and
\[
w_{2}(t,\xi)=\sqrt{C(\xi)}\left(  \sqrt{\widetilde{\varepsilon}(\xi
)}\widetilde{E}_{3}(t,\xi)+\sqrt{\mu}\widetilde{H}_{2}(t,\xi)j\right)
\]
where tilde means that the corresponding original function was written as a
function of $\xi$, e.g., $\widetilde{\varepsilon}(\xi(x))=\varepsilon(x)$.

Thus, all theory developed in \cite{KRT} and \cite{APFT} is applicable in this
case to the hyperbolic pseudoanalytic functions $w_{1}$ and $w_{2}$. As an
interesting application let us mention a possibility to construct an infinite
system of exact solutions of (\ref{BicompW}) using the results from \cite{KRT}
and \cite{APFT} combined with the elegant formulas obtained by Bers and
Gelbart in the elliptic case \cite{bers}. Namely, we notice that the pair of
functions
\begin{equation}
(F,G)=(f,j/f)\label{GenPair}%
\end{equation}
is a generating pair for both equations (\ref{w1w2}) and moreover it has a
form convenient for applying the results of Bers and Gelbart. Following
\cite{bers} (see \cite[Sect. 4.2]{APFT} for the slightly corrected formulas),
we give explicit formulas for the formal powers corresponding to the
generating pair (\ref{GenPair}) with $f$ depending on one Cartesian variable
$\xi$. For simplicity we assume that $z_{0}=0$ and $F(0)=1$. In this case the
formal powers are constructed in an elegant manner as follows. First, denote%

\[
X^{(0)}(\xi)=\widetilde{X}^{(0)}(\xi)=1
\]
and for $n=1,2,...$denote%

\[
X^{(n)}(\xi)=\left\{
\begin{tabular}
[c]{ll}%
$n%
{\displaystyle\int\limits_{0}^{\xi}}
X^{(n-1)}(x)\frac{dx}{f^{2}(x)}$ & $\text{for an odd }n$\\
$n%
{\displaystyle\int\limits_{0}^{\xi}}
X^{(n-1)}(x)f^{2}(x)dx$ & $\text{for an even }n$%
\end{tabular}
\ \right.
\]

\[
\widetilde{X}^{(n)}(\xi)=\left\{
\begin{tabular}
[c]{ll}%
$n%
{\displaystyle\int\limits_{0}^{\xi}}
\widetilde{X}^{(n-1)}(x)f^{2}(x)dx$ & $\text{for an odd }n$\\
$n%
{\displaystyle\int\limits_{0}^{\xi}}
\widetilde{X}^{(n-1)}(x)\frac{dx}{f^{2}(x)}$ & $\text{for an even }n$%
\end{tabular}
\ \ \right.
\]

Then for $a=a^{\prime}+ja^{\prime\prime}$, $a^{\prime},a^{\prime\prime}%
\in\mathbb{R}$ and $z=\xi+tj$ we have%
\[
Z^{(n)}(a,0,z)=f(\xi)\operatorname*{Re}\,_{\ast}Z_{{}}^{(n)}(a,0,z)+\frac
{j}{f(\xi)}\operatorname*{Im}\,_{\ast}Z_{{}}^{(n)}(a,0,z)
\]
where%

\begin{align}
_{\ast}Z^{(n)}(a,0,z)  &  =a^{\prime}%
{\displaystyle\sum\limits_{m=0}^{n}}
\binom{n}{m}X^{(n-m)}j^{m}t^{m}\text{\ }\label{Znodd}\\
&  +ja^{\prime\prime}%
{\displaystyle\sum\limits_{m=0}^{n}}
\binom{n}{m}\widetilde{X}^{(n-m)}j^{m}t^{m}\text{\ \ \ }\quad\text{for an odd
}n\nonumber
\end{align}
and%

\begin{align}
_{\ast}Z^{(n)}(a,0,z)  &  =a^{\prime}%
{\displaystyle\sum\limits_{m=0}^{n}}
\binom{n}{m}\widetilde{X}^{(n-m)}j^{m}t^{m}\text{\ }\label{Zneven}\\
&  +ja^{\prime\prime}%
{\displaystyle\sum\limits_{j=0}^{n}}
\binom{n}{m}X^{(n-m)}j^{m}t^{m}\text{\ \ \ \ }\quad\text{for an even
}n.\nonumber
\end{align}
For any\ $a\in\mathcal{H}$ and $n\in\mathbb{N}$ the formal power
$Z^{(n)}(a,0,z)$ is a solution of (\ref{w1w2}). Thus, the system of
constructed formal powers gives an infinite system of solutions of the Maxwell
equations in the case under consideration. Up to now it is an open question
what part of the kernel of (\ref{w1w2}) and consequently of the Maxwell system
(\ref{MaxOne2}) can be approximated by the obtained exact solutions.

\section{Generating solution for force-free magnetic fields}

The system describing force-free magnetic fields has the form%
\begin{gather}
\operatorname*{rot}\mathbf{B}+\alpha\mathbf{B}=0,\label{ff01}\\
\operatorname*{div}\mathbf{B}=0,\nonumber
\end{gather}
where $\mathbf{B}$ is a vector describing the magnetic field and the commonly
known as proportionality factor $\alpha$ is a real-valued function of spatial
variables. This system is of great importance in such fields as high
temperature superconductors and dynamics of magnetofluids (see e.g.
\cite{kamien}, \cite{marsh}, \cite{rcg}). It has been analized by different
methods including those of quaternionic analysis (see \cite{g-s},
\cite{KSbook} for the theory in the case when $\alpha$ is a constant and
\cite{kbook}, \cite{k-ff} and \cite{rgm-ff} for some developments in a
nonconstant case).

Obviously system (\ref{ff01}) can be written in the form of a quaternionic
equation%
\[
\left(  D+\alpha\right)  \mathbf{B}=0
\]
which is completely equivalent to (\ref{ff01}). In what follows we do not
restrict ourselves by purely vectorial null solutions of the operator
$D+\alpha$ and consider the equation
\begin{equation}
\left(  D+\alpha\right)  B=0\label{ff02}%
\end{equation}
on the class of continuously differentiable $\mathbb{H}(\mathbb{C})$-valued
functions $B$ with the proportionality factor $\alpha$ being a complex-valued
function. Applying the same idea as before we find out that for any scalar
(complex-valued) function $\varphi\in C^{1}(\Omega)$ where $\Omega$ is some
domain in $\mathbb{R}^{3}$ the equality
\begin{equation}
\left(  D+\alpha\right)  \left(  \varphi b\right)  =\left(  D\varphi\right)
b\label{ff03}%
\end{equation}
holds with $b$ being an $\mathbb{H}(\mathbb{C})$-valued function iff $b$ is a
solution of (\ref{ff02}). Thus, chosing four independent solutions of
(\ref{ff02}) $b_{k}$, $k=\overline{0,3}$ we can look for a general solution of
(\ref{ff02}) in the form $B=\sum_{k=0}^{3}\varphi_{k}b_{k}$ where the scalar
functions $\varphi_{k}$ should satisfy the equation
\[
\sum_{k=0}^{3}D\varphi_{k}\cdot b_{k}=0
\]
which is an analogue of the Bers equation for pseudoanalytic functions of the
second kind.

Nevertheless one can notice that in the case of equation (\ref{ff02}) in fact
it is sufficient to have only one particular solution $b$. As $b\lambda$ with
$\lambda$ being a constant biquaternion is also a solution of (\ref{ff02}) the
generating quartet $b_{k}$, $k=\overline{0,3}$ can be proposed in the form
$b_{k}=be_{k}$ if only $b$ is invertible in any point of the domain of
interest $\Omega$.

Thus, we obtain the following statement.

\begin{proposition}
Let $b$ be a particular solution of (\ref{ff02}), invertible in any point of
$\Omega\subset\mathbb{R}^{3}$. Then the general solution of (\ref{ff02}) can
be represented as the product $B=b\Phi$ where the components of the
$\mathbb{H}(\mathbb{C})$-valued function $\Phi$ satisfy the equation
\begin{equation}
\sum_{k=0}^{3}D\Phi_{k}\cdot be_{k}=0. \label{ff04}%
\end{equation}

\end{proposition}

\begin{proof}
Under the conditions of the proposition let us consider the function
$B=b\Phi=\sum_{k=0}^{3}b\Phi_{k}e_{k}$. Now application of the operator
$D+\alpha$ to $B$ gives $(D+\alpha)B=\sum_{k=0}^{3}D\Phi_{k}\cdot be_{k}%
+\sum_{k=0}^{3}\Phi_{k}\cdot\left(  (D+\alpha)b\right)  e_{k}$ from where it
is seen that $B$ is a solution of (\ref{ff02}) iff (\ref{ff04}) is valid.
\end{proof}

\begin{remark}
It should be noticed that the last proposition is also valid when $\alpha$ is
a full biquaternionic function.
\end{remark}

From the above proposition the following interesting fact about the quotients
of solutions of (\ref{ff02}) follows.

\begin{proposition}
\label{PropQuotient}Let $f$ and $g$ be solutions of (\ref{ff02}) and $f$ be
invertible in the domain of interest. Then the function $\Phi=f^{-1}g$ is a
solution of the equation
\begin{equation}
\sum_{k=0}^{3}D\Phi_{k}\cdot fe_{k}=0 \label{ff05}%
\end{equation}
and if $g$ is also invertible the inverse $\Psi=\Phi^{-1}=g^{-1}f$ is a
solution of the equation%
\begin{equation}
\sum_{k=0}^{3}D\Psi_{k}\cdot ge_{k}=0. \label{ff06}%
\end{equation}

\end{proposition}

\begin{proof}
According to the previous proposition if $g$ is a solution of (\ref{ff02}) it
can be represented in the form $g=f\Phi$ where the components of the
biquaternionic function $\Phi=f^{-1}g$ are solutions of (\ref{ff05}). Now
noticing that if $g$ is invertible as well the function $\Phi$ is also
invertible and hence consideration of $f=g\Phi^{-1}$ leads to equation
(\ref{ff06}).
\end{proof}

\begin{remark}
This proposition is of course also valid in the special case $\alpha\equiv0$
corresponding to quaternionic monogenic or hyperholomorphic\ functions. It is
well known that in general a quotient of two monogenic functions must not be
monogenic. Proposition \ref{PropQuotient} gives us a precise equation
satisfied by the quotient.
\end{remark}

\section{The Dirac equation}

We consider here the Dirac equation for one spin 1/2 particle under the
influence of an electromagnetic potential, but in fact the procedure is
applicable to other kinds of physical potentials (scalar, pseudoscalar, etc.).
The Dirac equation has the form
\begin{equation}
\left(  \gamma_{0}\partial_{t}-\sum_{k=1}^{3}\gamma_{k}\partial_{k}%
+im+i\mathbf{\phi}\gamma_{0}+i\sum_{k=1}^{3}A_{k}\gamma_{k}\right)  \Psi=0,
\label{d01}%
\end{equation}
where $\gamma_{0},\gamma_{1},\gamma_{2}$ and $\gamma_{3}$ are the Dirac
$\gamma$-matrices (see, e.g., \cite{Thaller}), $m$ is the mass of the
particle, $\phi$ is the electric potential and $A_{1},A_{2}$ and $A_{3}$ are
components of the magnetic potential $\overrightarrow{A}$. The wave function
$\Psi$ is a $\mathbb{C}^{4}$-vector function $\Psi=\left(  \Psi_{0},\Psi
_{1},\Psi_{2},\Psi_{3}\right)  $. Equation (\ref{d01}) is considered in some
domain $\mathcal{G}\subset\mathbb{R}^{4}$.

As was shown in \cite{vlad-bag} (see also \cite{KSbook} and \cite{kbook}) the
Dirac equation (\ref{d01}) can be written in the following quaternionic form%
\[
(D-\partial_{t}M^{e_{1}}+\mathbf{a}-M^{i\mathbf{\phi}e_{1}+me_{2}})\Phi=0
\]
where the purely vectorial quaternion $\mathbf{a}$ is obtained from the
magnetic potential $\overrightarrow{A}$, and the $\mathbb{H}(\mathbb{C}%
)$-valued function $\Phi$ is related to $\Psi$ by an invertible matrix
transformation (see \cite{KSbook}, \cite{kbook} and \cite{APFT}). It is worth
mentioning that this form of the Dirac equation was recently rediscovered in
\cite{Schwartz}.

Consideration of solutions of \ (\ref{d01}) with fixed energy $\Psi
(t,\mathbf{x})=\Psi_{\omega}(\mathbf{x})e^{i\omega t}$ leads to the
biquaternionic equation
\begin{equation}
(D+\mathbf{a}+M^{\mathbf{b}})W=0 \label{d04}%
\end{equation}
where $W$ is an $\mathbb{H}(\mathbb{C})$-valued function of three spatial
variables and $\mathbf{b=}-i(\mathbf{\phi}+\omega)e_{1}-me_{2}$.

\begin{proposition}
Let the biquaternionic functions $F_{0},F_{1},F_{2}$ and $F_{3}$ be
independent\footnote{In the sense that any $\mathbb{H}(\mathbb{C})$-valued
function $W$ can be written as $W=\sum_{k=0}^{3}\varphi_{k}F_{k}$ with
$\varphi_{k}$ being scalar functions.} solutions of (\ref{d04}). Then
$W=\sum_{k=0}^{3}\varphi_{k}F_{k}$ where $\varphi_{0},\varphi_{1},\varphi_{2}$
and $\varphi_{3}$ are scalar functions is a solution of (\ref{d04}) if and
only if the following equation is satisfied
\begin{equation}
\sum_{k=0}^{3}(D\varphi_{k})F_{k}=0. \label{d05}%
\end{equation}

\end{proposition}

Obviously, equation (\ref{d05}) is an analogue of equation (\ref{ps02})
describing pseudoanalytic functions of the second kind.

\section{Conclusions}

In the present paper it was shown that the concept of a generating pair is not
limited to the classical pseudoanalytic function theory and can be introduced
in relation with a variety of first-order systems of mathematical physics.
Here we considered the Maxwell system for inhomogeneous media, the system
describing force-free magnetic fields and the Dirac system from relativistic
quantum mechanics. Nevertheless the approach presented in this paper and based
on the consideration of a generating solution as an intertwining operator is
clearly more general. The knowledge of a generating set of solutions makes it
possible to rewrite the original system in a form analogous to the Vekua
equation for pseudoanalytic functions of the second kind which opens the way
to construction of classes of solutions as was shown in section 3 for the
Maxwell system and perhaps more important to the development of Bers' theory
for the corresponding system of mathematical physics including such concepts
as the derivative, antiderivative, formal powers, Taylor and Laurent series, etc.

\end{document}